\documentclass[envcountsame]{llncs}
\newcommand{\lv}[1]{#1}
\newcommand{\sv}[1]{}

\usepackage{complexity}
\usepackage{graphicx}
\usepackage{framed}
\usepackage{amsmath}
\usepackage{amsfonts}
\usepackage{mathtools}
\usepackage{cite}

\let\doendproof\endproof
\renewcommand\endproof{~\hfill\qed\doendproof}

\DeclareMathOperator{\diam}{diam}
\DeclareMathOperator{\rc}{rc}
\DeclareMathOperator{\src}{src}
\DeclareMathOperator{\rvc}{rvc}
\DeclareMathOperator{\srvc}{srvc}
\DeclareMathOperator{\vcn}{vcn}
\DeclareMathOperator{\tw}{tw}

\DeclareMathOperator{\inc}{I}

\DeclareMathOperator{\path}{Path}
\DeclareMathOperator{\vertexconnects}{VertexConnects}
\DeclareMathOperator{\edgeconnects}{EdgeConnects}
\DeclareMathOperator{\rainbow}{Rainbow}

\newlang{\subsetsrvc}{\emph{k}\text{-}SSRVC}
\newlang{\rcprob}{RC}
\newlang{\rvcprob}{RVC}
\newlang{\srcprob}{SRC}
\newlang{\srvcprob}{SRVC}
\newlang{\savingrc}{SavingRC}
\newlang{\savingrvc}{SavingRVC}

\newcommand{\bigoh}{\mathcal{O}}
\newcommand{\Nat}{\mathbb{N}}
\usepackage{todonotes}

\spnewtheorem{ourfact}[theorem]{Fact}{\bfseries}{\itshape}
\newcommand{\SB}{\{\,} \newcommand{\SM}{\;{|}\;} \renewcommand{\SE}{\,\}}

\begin{document}

\title{On the Complexity of Rainbow Coloring Problems}

\author{
Eduard Eiben\inst{1} \and
Robert Ganian\inst{1} \and
Juho Lauri\inst{2}
}

\institute{
TU Wien, Vienna, Austria\\\email{eiben@ac.tuwien.ac.at\\ rganian@gmail.com} \vspace{0.4em} \and
Tampere University of Technology, Tampere, Finland\\\email{juho.lauri@tut.fi}}
\maketitle

\begin{abstract}
An edge-colored graph $G$ is said to be \emph{rainbow connected} if between each pair of vertices there exists a path which uses each color at most once. The \emph{rainbow connection number}, denoted by $\rc(G)$, is the minimum number of colors needed to make $G$ rainbow connected. Along with its variants, which consider vertex colorings and/or so-called strong colorings, the rainbow connection number has been studied from both the algorithmic and graph-theoretic points of view. 

In this paper we present a range of new results on the computational complexity of computing the four major variants of the rainbow connection number. In particular, we prove that the  \textsc{Strong Rainbow Vertex Coloring} problem is $\NP$-complete even on graphs of diameter $3$. We show that when the number of colors is fixed, then all of the considered problems can be solved in linear time on graphs of bounded treewidth. Moreover, we provide a linear-time algorithm which decides whether it is possible to obtain a rainbow coloring by saving a fixed number of colors from a trivial upper bound.
Finally, we give a linear-time algorithm for computing the exact rainbow connection numbers for three variants of the problem on graphs of bounded vertex cover number.
\end{abstract}

\section{Introduction}
\label{sec:Introduction}
The concept of rainbow connectivity was introduced by Chartrand, Johns, McKeon, and Zhang in 2008~\cite{Chartrand2008} as an interesting connectivity measure motivated by recent developments in the area of secure data transfer. Over the past years, this strengthened notion of connectivity has received a significant amount of attention in the research community. The applications of rainbow connectivity are discussed in detail for instance in the recent survey~\cite{Li2012}, and various bounds are also available in~\cite{Chartrand2008b, Li2012b}. 

An edge-colored graph $G$ is said to be \emph{rainbow connected} if between each pair of vertices $a,b$ there exists an $a-b$ path which uses each color at most once; such a path is called \emph{rainbow}. The minimum number of colors needed to make $G$ rainbow connected is called the \emph{rainbow connection number} ($\rc$), and the \textsc{Rainbow Coloring} problem asks to decide if the rainbow connection number is upper-bounded by a number specified in the input. Precise definitions are given in Section~\ref{sec:prelim}.

The rainbow connection number and \textsc{Rainbow Coloring} have been studied from both the algorithmic and graph-theoretic points of view. On one hand, the exact rainbow connection numbers are known for a variety of simple graph classes, such as wheel graphs~\cite{Chartrand2008}, complete multipartite graphs~\cite{Chartrand2008}, unit interval graphs~\cite{Chandran2012b}, and threshold graphs~\cite{Chandran2012}. On the other hand, \textsc{Rainbow Coloring} is a notoriously hard problem. It was shown by Chakraborty {\em et al.}~\cite{Chakraborty2009} that already deciding if $\rc(G) \leq 2$ is $\NP$-complete, and Ananth {\em et al.}~\cite{Ananth2011} showed that for any $k>2$ deciding $\rc(G)\leq k$ is $\NP$-complete. In fact, Chandran and Rajendraprasad~\cite{Chandran2012} strengthened this result to hold for chordal graphs. In the same paper, the authors gave a linear time algorithm for rainbow coloring split graphs which form a subclass of chordal graphs with at most one more color than the optimum. 
\lv{Basavaraju {\em et al.}~\cite{Basavaraju2012} gave an $(r+3)$-factor approximation algorithm to rainbow color a general graph of radius $r$. }
Later on, the inapproximability of the problem was investigated by Chandran and Rajendraprasad~\cite{Chandran2013}. They proved that there is no polynomial time algorithm to rainbow color graphs with less than twice the minimum number of colors, unless $\P = \NP$. 
\lv{For chordal graphs, they gave a $5/2$-factor approximation algorithm, and proved that it is impossible to do better than $5/4$, unless $\P = \NP$.}

Several variants of the notion of rainbow connectivity have also been considered. Indeed, a similar concept was introduced for vertex-colored graphs by Krivelevich and Yuster~\cite{Krivelevich2010}. A vertex-colored graph $H$ is \emph{rainbow vertex connected} if there is a path whose internal vertices have distinct colors between every pair of vertices, and this gives rise to the \emph{rainbow vertex connection number} ($\rvc$). The \emph{strong rainbow connection number} ($\src$) was introduced and investigated also by Chartrand {\em et al.}~\cite{Chartrand2008b}; an edge-colored graph $G$ is said to be \emph{strong rainbow connected} if between each pair of vertices $a,b$ there exists a shortest $a-b$ path which is rainbow. The combination of these two notions, \emph{strong rainbow vertex connectivity} ($\srvc$), was studied in a graph theoretic setting by Li {\em et al.}~\cite{Li2014}.

Not surprisingly, the problems arising from the strong and vertex variants of rainbow connectivity are also hard. Chartrand {\em et al.} showed that $\rc(G) = 2$ if and only if $\src(G) = 2$~\cite{Chartrand2008}, and hence deciding if $\src(G) \leq k$ is $\NP$-complete for $k=2$. The problem remains $\NP$-complete for $k > 2$ for bipartite graphs~\cite{Ananth2011}, and also for split graphs~\cite{Keranen2014}. Furthermore, the strong rainbow connection number of an $n$-vertex bipartite graph cannot be approximated within a factor of $n^{1/2-\epsilon}$, where $\epsilon > 0$ unless $\NP = \ZPP$~\cite{Ananth2011}, and the same holds for split graphs~\cite{Keranen2014}. The computational aspects of the rainbow vertex connection numbers have received less attention in the literature. Through the work of Chen {\em et al.}~\cite{Chen2011} and Chen {\em et al.}~\cite{Chen2013}, it is known that deciding if $\rvc(G) \leq k$ is $\NP$-complete for every $k \geq 2$. However, to the best of our knowledge, the complexity of deciding whether $\srvc(G)\leq k$ (the $k$-\srvcprob{} problem) has not been previously considered.

In this paper, we present new positive and negative results for all four variants of the rainbow coloring problems discussed above.
\begin{itemize}
\item In Section~\ref{sec:hardness_srvc}, we prove that $k$-\srvcprob{} is $\NP$-complete for every $k\geq 3$ even on graphs of diameter $3$. Our reduction relies on an intermediate step which proves the $\NP$-hardness of a more general problem, the \textsc{$k$-Subset Strong Rainbow Vertex Coloring} problem. We also provide bounds for approximation algorithms (under established complexity assumptions), see Corollary~\ref{cor:noapprox}. 
\item In Section~\ref{sec:mso}, we show that all of the considered problems can be formulated in monadic second order (MSO) logic. In particular, this implies that for every fixed $k$, all of the considered problems can be solved in linear time on graphs of bounded treewidth, and the vertex variants can be solved in cubic time on graphs of bounded clique-width. 
\item In Section~\ref{sec:saving}, we investigate the problem from a different perspective: we ask whether, given an $n$-vertex graph $G$ and an integer $k$, it is possible to color $G$ using $k$ colors less than the known upper bound. Here we employ a win-win approach and show that this problem can be solved in time $\bigoh(n)$ for any fixed $k$. 
\item In the final Section~\ref{sec:vc}, we show that in the general case when $k$ is not fixed, three of the considered problems admit linear-time algorithms on graphs of bounded vertex cover number. This is also achieved by exploiting a win-win approach, where we show that either $k$ is bounded by a function of the vertex cover number and hence we can apply the result of Section~\ref{sec:mso}, or $k$ is sufficiently large which allows us to exploit the structure of the graph to precisely compute the connectivity number.	
\end{itemize}

\smallskip \noindent {\emph{The authors acknowledge support by the Austrian Science Fund (FWF, projects P26696 and W1255-N23).}}

\section{Preliminaries}
\label{sec:prelim}

\subsection{Graphs and Rainbow Connectivity}
We refer to~\cite{Diestel2010} for standard graph-theoretic notions. We use $[i]$ to denote the set $\{1,2,\dots,i\}$. All graphs considered in this paper are simple and undirected. The \textit{degree} of a vertex is the number of its incident edges, and a vertex is a \emph{pendant} if it has degree $1$. We will often use the shorthand $ab$ for the edge $\{a,b\}$. For a vertex set $X$, we use $G[X]$ to denote the subgraph of $G$ induced on $X$.

A \emph{vertex coloring} of a graph $G=(V,E)$ is a mapping from $V$ to $\Nat$, and similarly an \emph{edge coloring} of $G$ is a mapping from $E$ to $\Nat$; in this context, we will often refer to the elements of $\Nat$ as colors. 
An $a-b$ \emph{path} $P$ \emph{of length} $p$ is a finite sequence of the form $(a=v_0,e_0,v_1,e_1,\dots b=v_p)$, where $v_0,v_1,\dots v_p$ are distinct vertices and $e_0,\dots e_{p-1}$ are distinct edges and each edge $e_j$ is incident to $v_j$ and $v_{j+1}$. An $a-b$ path of length $p$ is a \emph{shortest path} if every $a-b$ path has length at least $p$. The \textit{diameter} of a graph $G$ is the length of its longest shortest path, denoted by $\diam(G)$. Given an edge (vertex) coloring $\alpha$ of $G$, a color $x\in \Nat$ \emph{occurs} on a path $P$ if there exists an edge (an internal vertex) $z$ on $P$ such that $\alpha(z)=x$.

A vertex or edge coloring of $G$ is \emph{rainbow} if between each pair of vertices $a,b$ there exists an $a-b$ path $P$ such that each color occurs at most once on $P$; in this case we say that $G$ is \emph{rainbow connected} or \emph{rainbow colored}.
We denote by $\rc(G)$ the minimum $i\in \Nat$ such that there exists a rainbow edge coloring $\alpha:E\rightarrow [i]$. Similarly, $\rvc(G)$ denotes the minimum $i\in \Nat$ such that there exists a rainbow vertex coloring $\alpha:V\rightarrow [i]$.
Furthermore, an edge or vertex coloring of $G$ is a \emph{strong rainbow coloring} if between each pair of vertices $a,b$ there exists a shortest $a-b$ path $P$ such that each color occurs at most once on $P$. We denote by $\src(G)$ ($\srvc(G)$) the minimum $i\in \Nat$ such that there exists a strong rainbow edge (vertex) coloring $\alpha:E\rightarrow [i]$ ($\alpha:V\rightarrow [i]$).

Let $G$ and $H$ be two graphs with $n$ and $n'$ vertices, respectively. The \textit{corona} of $G$ and $H$, denoted by $G \circ H$, is the disjoint union of $G$ and $n$ copies of $H$ where the $i$-th vertex of $G$ is connected by an edge to every vertex of the $i$-th copy of $H$. Clearly, the corona $G \circ H$ has $n(1+n')$ vertices. \lv{Coronas of graphs were first studied by Frucht and Harary~\cite{Frucht1970}.}
  
\subsection{Problem Statements}
\label{sub:prob}
Here we formally state the problems studied in this work.

\begin{framed}
\vspace{-0.2cm}
\noindent \textsc{Rainbow $k$-Coloring ($k$-\rcprob)} \\
\textbf{Instance:} A connected undirected graph $G=(V,E)$. \\ 
\textbf{Question:} Is $\rc(G) \leq k$?
\vspace{-0.2cm}
\end{framed}

\textsc{Strong Rainbow $k$-Coloring ($k$-\srcprob)}, \textsc{Rainbow Vertex $k$-Coloring ($k$-\rvcprob)} and \textsc{Strong Rainbow Vertex $k$-Coloring ($k$-\srvcprob)} are then defined analogously for $\src(G)$, $\rvc(G)$ and $\srvc(G)$, respectively.
We also consider generalized versions of these problems, where $k$ is given as part of the input.

\begin{framed}
\vspace{-0.2cm}
\noindent \textsc{Rainbow Coloring (\rcprob)} \\
\textbf{Instance:} A connected undirected graph $G=(V,E)$, and a positive integer~$k$. \\ 
\textbf{Question:} Is $\rc(G) \leq k$?
\vspace{-0.2cm}
\end{framed}

The problems $\srcprob$, $\rvcprob$, and $\srvcprob$ are also defined analogously. In Section~\ref{sec:saving} we consider the ``saving'' versions of the problem, which ask whether it is possible to improve upon the trivial upper bound for the number of colors.

\begin{framed}
\vspace{-0.2cm}
\noindent \textsc{Saving $k$ Rainbow Colors ($k$-\savingrc)} \\
\textbf{Instance:} A connected undirected graph $G=(V,E)$. \\ 
\textbf{Question:} Is $\rc(G) \leq |E|-k$?\medskip
\medskip

\noindent \textsc{Saving $k$ Rainbow Vertex Colors ($k$-\savingrvc)} \\
\textbf{Instance:} A connected undirected graph $G=(V,E)$. \\ 
\textbf{Question:} Is $\rvc(G) \leq |V|-k$?
\vspace{-0.2cm}
\end{framed}


\subsection{Structural Measures}
Several of our results utilize certain structural measures of graphs. We will mostly be concerned with the \emph{treewidth} and the \emph{vertex cover number} of the input graph. Section~\ref{sec:mso} also mentions certain implications of our results for graphs of bounded \emph{clique-width}, the definition of which can be found for instance in~\cite{CourcelleMakowskyRotics00}.

A \emph{tree decomposition} of $G$ is a pair 
$(T,\{X_i : i\in I\})$
where $X_i \subseteq V$, $i\in I$, and $T$ is a tree with elements
of $I$ as nodes
such that:
\begin{enumerate}
\item for each edge $uv\in E$, there is an $i\in I$ such that $\{u,v\} 
\subseteq X_i$, and
\item for each vertex $v\in V$, $T[\SB i\in I \SM v\in X_i \SE]$ is a (connected) tree with at least one node.
\end{enumerate}
The \emph{width} of a tree decomposition is $\max_{i \in I} |X_i|-1$.
The \emph{treewidth}~\cite{Robertson1986} of $G$
is the minimum width taken over all tree decompositions
of $G$ and it is denoted by $\tw(G)$. 


\begin{ourfact}[\hspace{1sp}\cite{Bodlaender1996}]
\label{fact:findtw}
There exists an algorithm which, given a graph $G$ and an integer $p$,
runs in time $2^{p^{\bigoh(1)}}\cdot (|V(G)|+|E(G)|)$, and either
outputs a tree decomposition of $G$ of width at most $p$ or correctly
determines that $\tw(G)>p$.
\end{ourfact}

A \emph{vertex cover} of a graph $G=(V,E)$ is a set $X\subseteq V$ such that each edge in $G$ has at least one endvertex in $X$. The cardinality of a minimum vertex cover in $G$ is denoted as $\vcn(G)$. Given a vertex cover $X$, a \emph{type} $T$ is a subset of $V\setminus X$ such that any two vertices in $T$ have the same neighborhood; observe that any graph contains at most $2^{|X|}$ many distinct types.


\subsection{Monadic Second Order Logic}
We assume that we have an infinite supply of individual variables,
denoted by lowercase letters $x,y,z$, and an infinite supply of set
variables, denoted by uppercase letters $X,Y,Z$. \emph{Formulas} of
MSO$_2$ logic are constructed from atomic
formulas $I(x,y)$, $x\in X$, and $x = y$ using the connectives $\neg$
(negation), $\wedge$ (conjunction) and existential quantification
$\exists x$ over individual variables as well as existential
quantification $\exists X$ over set variables. Individual variables
range over vertices and edges, and set variables range either over sets of
vertices or over sets of edges. The atomic formula $I(x,y)$ expresses that vertex $x$ is incident to edge $y$, $x = y$
expresses equality, and $x\in X$ expresses that $x$ is in the set
$X$. \lv{From this, we define the semantics of MSO$_2$ logic
in the standard way.}

MSO$_1$ logic is defined similarly as MSO$_2$ logic, with the following distinctions. Individual variables range only over vertices, and set variables only range over sets of vertices. The atomic formula $I(x,y)$ is replaced by $E(x,y)$, which expresses that vertex $x$ is adjacent to vertex $y$.

\emph{Free and bound variables} of a formula are defined in the usual way. A
\emph{sentence} is a formula without free variables. It is known that MSO$_2$ formulas can be checked efficiently as long as the graph has bounded tree-width. 


\begin{ourfact}[\hspace{1sp}\cite{Courcelle90a}]\label{fact:msotreewidth}
Let $\phi$ be a fixed \emph{MSO}$_2$ sentence and $p\in \Nat$ be a constant. Given an $n$-vertex graph $G$ of treewidth at most $p$, it is possible to decide whether $G\models \phi$ in time $\bigoh(n)$.
\end{ourfact}

Similarly, MSO$_1$ formulas can be checked efficiently as long as the graph has bounded \emph{clique-width}~\cite{CourcelleMakowskyRotics00} (or, equivalently, \emph{rank-width}~\cite{GanianHlineny10}). In particular, while the formula can be checked in linear time if a suitable rank- or clique-decomposition is provided, current algorithms for finding (or approximating) such a decomposition require cubic time.


\begin{ourfact}[\hspace{1sp}\cite{CourcelleMakowskyRotics00,GanianHlineny10}]
\label{fact:msorankwidth}
Let $\phi$ be a fixed \emph{MSO}$_1$ sentence and $p\in \Nat$ be a constant. Given an $n$-vertex graph $G$ of clique-width at most $p$, it is possible to decide whether $G\models \phi$ in time $\bigoh(n^3)$.
\end{ourfact}

\section{Hardness of Strong Rainbow Vertex $k$-Coloring}
\label{sec:hardness_srvc}
It is easy to see that $\srvc(G) = 1$ if and only if $\diam(G) = 2$. We will prove that deciding if $\srvc(G) \leq k$ is $\NP$-complete for every $k \geq 3$ already for graphs of diameter 3. This is done by first showing hardness of an intermediate problem, described below. 

In the \textsc{$k$-Subset Strong Rainbow Vertex Coloring} problem ($\subsetsrvc$) we are given a graph $G$ which is a corona of a complete graph and $K_1$, and a set $P$ of pairs of pendants in $G$. The goal is to decide if the vertices of $G$ can be colored with $k$ colors such that each pair in $P$ is connected by a vertex rainbow shortest path. We will first show this intermediate problem is $\NP$-complete by reducing from the classical $k$-vertex coloring problem: given a graph $G$, decide if there is an assignment of $k$ colors to the vertices of $G$ such that adjacent vertices receive a different color. The $k$-vertex coloring problem is well-known to be $\NP$-complete for every $k \geq 3$.

\sv{\begin{lemma}[$\star$]}
\lv{\begin{lemma}}
\label{lem:ssrvc}
The $\subsetsrvc$ problem is $\NP$-complete for every $k \geq 3$.
\end{lemma}

\newcommand{\pfssrvc}[0]{
\begin{proof}
Let $G=(V,E)$ be an instance of the $k$-vertex coloring problem, where $k \geq 3$. We will construct an instance $\langle G',P \rangle$ of the $\subsetsrvc$ problem such that $\langle G',P \rangle$ is a YES-instance if and only if $G$ is $k$-vertex colorable.

The graph $G' = (V',E')$ along with the set of pairs $P$ are constructed as follows:
\begin{itemize}
\item $V' = V \cup \{ p_v \mid v \in V \}$,
\item $E' = \{ uv \mid u,v \in V \wedge u \neq v \} \cup \{ vp_v \mid v \in V \}$, and
\item $P = \{ \{p_u,p_v\} \mid uv \in E \}$.
\end{itemize}
Clearly, $G' = K_{|V|} \circ K_1$. This completes the construction of $G'$.

Suppose $G$ is $k$-vertex colorable. Let $c$ be the color assigned to vertex $v$ in $V$. We assign the color $c$ to both $v$ and $p_v$ in $G'$. Observe that the shortest path between any pair of vertices in $G'$ is unique. It is then straightforward to verify that any pair in $P$ is strong rainbow vertex connected.

For the other direction, suppose there is a vertex coloring of $G'$ using $k$ colors under which there is a vertex rainbow shortest path between every pair in $P$. Since any two vertices $\{p_u,p_v\} \in P$ are strong rainbow vertex connected, the two internal vertices on the unique $p_u-p_v$ shortest path have distinct colors. Thus by assigning to the vertex $v \in V$ the color on the corresponding vertex $v' \in (V' \setminus \{ p_v \mid v \in V \})$ we get a proper vertex coloring of $G$. This completes the proof.
\end{proof}
}
\lv{\pfssrvc}
\sv{\begin{proof}[Sketch]
Let $G=(V,E)$ be an instance of the $k$-vertex coloring problem, where $k \geq 3$. We will construct an instance $\langle G',P \rangle$ of the $\subsetsrvc$ problem such that $\langle G',P \rangle$ is a YES-instance if and only if $G$ is $k$-vertex colorable. 

The graph $G' = (V',E')$ along with the set of pairs $P$ are constructed as follows:
\begin{itemize}
\item $V' = V \cup \{ p_v \mid v \in V \}$,
\item $E' = \{ uv \mid u,v \in V \wedge u \neq v \} \cup \{ vp_v \mid v \in V \}$, and
\item $P = \{ \{p_u,p_v\} \mid uv \in E \}$.
\end{itemize}
Clearly, $G' = K_{|V|} \circ K_1$. This completes the construction of $G'$. Satisfying color assignments for $V$ in $G$ are necessarily satisfying color assignments for $V$ in $G'$, and vice versa.
\end{proof}
}

We are now ready to prove the following.
\lv{\begin{theorem}}
\sv{\begin{theorem}[$\star$]}
\label{thm_srvc_hardness}
The problem $k$-\srvcprob{} is NP-complete for every integer $k \geq 3$, even when the input is restricted to graphs of diameter $3$.
\end{theorem}

\newcommand{\pfsrvc}[0]{
\begin{proof}
Let $k\geq 3$ and $\langle G=(V,E),P \rangle$ be an instance of the $\subsetsrvc$ problem. We will construct a graph $G'=(V',E')$ that is strong rainbow vertex colorable with $k$ colors if and only if $\langle G=(V,E),P \rangle$ is a YES-instance of $\subsetsrvc$.

Let $V_1$ denote the set of pendant vertices in $G$. For every vertex $v \in V_1$ we introduce a new vertex $x_v$. For every pair of pendant vertices $\{u,v\} \not \in P$, we add two vertices $x^1_{uv}$ and $x^2_{uv}$. We also add two new vertices $s$ and $t$. In the following, we denote by $k_v$, where $v \in V_1$,  the unique vertex that $v$ is adjacent to in $G$. Formally, we construct a graph $G' = (V',E')$ such that:
\begin{itemize}
\item $V' = V \cup \{ x_v \mid v \in V_1 \} \cup \{ x^1_{uv},x^2_{uv} \mid \{u,v\} \in {V_1 \choose 2} \setminus P \} \cup \{ s,t \}$,
\item $E' = E \cup E_1 \cup E_2 \cup E_3 \cup E_4$,
\item $E_1 = \{ vx_v,sx_v,tx_v \mid v \in V_1 \}$,
\item $E_2 = \{ ux^1_{uv}, x^1_{uv}x^2_{uv},x^2_{uv}v \mid \{u,v\} \in {V_1 \choose 2} \setminus P \}$,
\item $E_3 = \{ sx^1_{uv},tx^2_{uv}, k_ux^1_{uv}, k_vx^2_{uv} \mid \{u,v\} \in {V_1 \choose 2} \setminus P \}$, and
\item $E_4 = \{ sy,ty \mid y \in V \setminus V_1 \}$.
\end{itemize}
This completes the construction of $G'$. It is straightforward to verify $\diam(G') = 3$, and this is realized between any pair of vertices in $V_1$. An example illustrating the reduction is shown in Figure~\ref{fig:reduction_srvc}. 

First, suppose $G'$ admits a strong rainbow vertex coloring $\phi$ using $k$ colors. Observe that for each $\{u,v\}\in P$, the shortest path between $u$ and $v$ in $G'$ is unique. Therefore $k_u$ and $k_v$ must receive distinct colors by $\phi$. Hence the restriction of $\phi$ to $V$ witnesses that $\langle G, P \rangle$ is a YES-instance of \subsetsrvc.


On the other hand, suppose $\langle G, P \rangle$ is $k$-subset strong rainbow vertex connected under some coloring $\phi : V \to \{c_1,\ldots,c_k\}$. We will describe an extended $k$-coloring $\phi'$ under which $G'$ is strong rainbow vertex connected. We retain the original coloring on the vertices of $G$, i.e., $\phi'(v) = \phi(v)$ for every $v \in V$. The remaining vertices in $G'$ receive colors as follows:
\begin{itemize}
\item $\phi'(x_v) = c_1$, for every $v \in V_1$,
\item $\phi'(x^1_{uv}) = c_1$, $\phi'(x^2_{uv}) = c_2$, for every $\{u,v\} \in {V_1 \choose 2} \setminus P$, and
\item $\phi'(s) = c_2$, and $\phi'(t) = c_2$.
\end{itemize}
Since each pair of vertices $\{a,b\}\in V'$ at distance at most $2$ are always rainbow connected regardless of the chosen coloring, and each pair of vertices $\{u,v\} \in {V_1 \choose 2} \setminus P$ are connected by the path through $x^1_{uv},x^2_{uv}$, it is straightforward to verify that $G'$ is indeed strong rainbow vertex connected under $\phi'$.
\end{proof}}
\lv{\pfsrvc}
\sv{\begin{proof}[Sketch]
Let $k\geq 3$ and $\langle G=(V,E),P \rangle$ be an instance of the $\subsetsrvc$ problem. We will construct a graph $G'=(V',E')$ that is strong rainbow vertex colorable with $k$ colors if and only if $\langle G=(V,E),P \rangle$ is a YES-instance of $\subsetsrvc$.

Let $V_1$ denote the set of pendant vertices in $G$. For every vertex $v \in V_1$ we introduce a new vertex $x_v$. For every pair of pendant vertices $\{u,v\} \not \in P$, we add two vertices $x^1_{uv}$ and $x^2_{uv}$. We also add two new vertices $s$ and $t$. In the following, we denote by $k_v$, where $v \in V_1$,  the unique vertex that $v$ is adjacent to in $G$. Formally, we construct a graph $G' = (V',E')$ such that:
\begin{itemize}
\item $V' = V \cup \{ x_v \mid v \in V_1 \} \cup \{ x^1_{uv},x^2_{uv} \mid \{u,v\} \in {V_1 \choose 2} \setminus P \} \cup \{ s,t \}$,
\item $E' = E \cup E_1 \cup E_2 \cup E_3 \cup E_4$,
\item $E_1 = \{ vx_v,sx_v,tx_v \mid v \in V_1 \}$,
\item $E_2 = \{ ux^1_{uv}, x^1_{uv}x^2_{uv},x^2_{uv}v \mid \{u,v\} \in {V_1 \choose 2} \setminus P \}$,
\item $E_3 = \{ sx^1_{uv},tx^2_{uv}, k_ux^1_{uv}, k_vx^2_{uv} \mid \{u,v\} \in {V_1 \choose 2} \setminus P \}$, and
\item $E_4 = \{ sy,ty \mid y \in V \setminus V_1 \}$.
\end{itemize}
This completes the construction of $G'$. It is easy to verify $\diam(G') = 3$.
\end{proof}
}

\begin{figure}[t]
\begin{center}
	\includegraphics[scale=0.7]{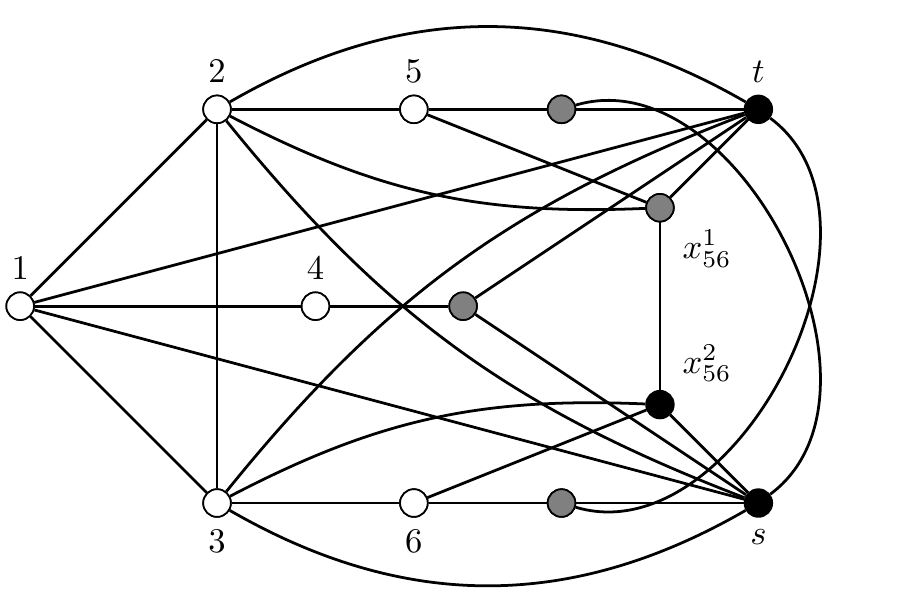}
\end{center}
\caption{The graph $K_3 \circ K_1$ transformed to a graph of diameter 3 with $P = \{ \{4,5\}, \{4,6\} \}$. The color $c_1$ is represented with grey, and the color $c_2$ with black. White vertices represent an unknown vertex coloring under which the pairs in $P$ are strong rainbow vertex connected.}
\vspace{-0.5cm}
\label{fig:reduction_srvc}
\end{figure}

It can be observed that the size of the above reduction does not depend on $k$, the number of colors. In fact, if the instance of the $k$-vertex coloring problem has $n$ vertices, then the graph $G'$ we build in Theorem~\ref{thm_srvc_hardness} has no more than $\bigoh(n^2)$ vertices. Furthermore, a strong rainbow vertex coloring of $G'$ gives us a solution to the $k$-vertex coloring problem. Since the vertex coloring number of an $n$-vertex graph cannot be approximated within a factor of $n^{1-\epsilon}$ for any $\epsilon > 0$ unless $\P = \NP$~\cite{Zuckerman2006}, we obtain the following corollary.

\begin{corollary}
\label{cor:noapprox}
There is no polynomial time algorithm for approximating the strong rainbow vertex connection number of an $n$-vertex graph of bounded diameter within a factor of $n^{1/2-\epsilon}$ for any $\epsilon > 0$, unless $\P = \NP$.
\end{corollary}

\section{MSO Formulations}
\label{sec:mso}

This section will present formulations of the $k$-coloring variants of rainbow connectivity in MSO logic, along with their algorithmic implications.

\begin{lemma}
\label{lem:MSOstandard}
For every $k\in \Nat$ there exists a \emph{MSO}$_1$ formula $\phi_k$ such that for every graph $G$, it holds that $G\models \phi$ iff $G$ is a YES-instance of $k$-\rvcprob. Similarly, for every $k\in \Nat$ there exists a \emph{MSO}$_2$ formula $\psi_k$ such that for every graph $G$, it holds that $G\models \psi$ iff $G$ is a YES-instance of $k$-\rcprob.
\end{lemma}

\begin{proof}
In the case of $k$-\rcprob, we wish to partition the edges of the graph $G=(V,E)$ into $k$ color classes $C_1,\ldots,C_k$ such that each pair of vertices is connected by a rainbow path. Let us consider the following MSO$_2$ formula $\psi_k$.
\begin{equation*}
\begin{split}
\psi_k 	&\coloneqq \exists C_1,\ldots,C_k \subseteq E \Big( \forall e \in E \Big( e \in C_1 \vee \cdots \vee e \in C_k \Big) \Big) \\
	&\wedge \Big( \forall i,j\in [k], i\neq j: (C_i \cap C_j = \emptyset) \Big) \\
	&\wedge \Big( \forall u,v \in V \Big( (u \neq v) \implies \bigvee_{1 \leq i \leq k} \Big( \exists e_1,\ldots,e_i \in E \big( \path(u,v,e_1,\ldots,e_i) \\
	&\wedge \rainbow(e_1,\ldots,e_i) \big) \Big) \Big) \Big), 
\end{split}
\end{equation*}
\lv{where the auxiliary predicates are defined as
\begin{equation*}
\begin{split}
\path(u,v,e_1,\ldots,e_\ell) &\coloneqq \exists v_1,\ldots,v_{\ell-1} \in V \\
	&\Big( \forall i,j \in [\ell-1], i\neq j: (v_i \neq v_j) \Big) \\
	&\wedge \inc(e_1,u) \wedge \inc(e_\ell,v)\\
&\wedge \Big( \forall i\in[\ell-1]: (\inc(e_i,v_i) \wedge \inc(e_{i+1},v_i)) \Big)
\end{split}
\end{equation*}
and
\begin{equation*}
\begin{split}
\rainbow(e_1,\ldots,e_\ell) &\coloneqq \forall i\in[\ell] ~ \exists j\in[k]: \Big( e_i \in C_j 
	\wedge ( \forall p\neq i: e_p \notin C_j ) \Big).
\end{split}
\end{equation*}}
Here, $\path(u,v,e_1,\ldots,e_\ell)$ expresses that the edges $e_1,\ldots,e_\ell$ form a path between the vertices $u$ and $v$. The predicate $\rainbow(e_1,\ldots,e_\ell)$ expresses that the edges $e_1,\ldots,e_\ell$ are each in precisely one color class. \sv{The formal definition of these predicates is available in the appendix.}

In the case of $k$-\rvcprob, the MSO$_1$ formula $\phi_k$ is defined analogously, with the following distinctions: 
\begin{enumerate}
\item instead of edges, we partition the vertices of $G$ into color classes; 
\item the predicate $\path$ speaks of vertices instead of edges and uses the adjacency relation instead of the incidence relation; and
\item the predicate $\rainbow$ tests the coloring of vertices instead of edges.
\end{enumerate}
\vspace{-0.65cm}
\end{proof}

Using a similar approach, we obtain an analogous result for the strong variants of these problems.

\lv{\begin{lemma}}
\sv{\begin{lemma}[$\star$]}
\label{lem:MSOstrong}
For every $k\in \Nat$ there exists a MSO$_1$ formula $\phi_k$ such that for every graph $G$, it holds that $G\models \phi$ iff $G$ is a YES-instance of $k$-\textsc{SRVC}. Similarly, for every $k\in \Nat$ there exists a MSO$_2$ formula $\psi_k$ such that for every graph $G$, it holds that $G\models \psi$ iff $G$ is a YES-instance of $k$-\srcprob.
\end{lemma}

\newcommand{\pfMSOstrong}[0]{
\begin{proof}
In the case of $k$-\srcprob, we wish to partition the edges of the graph $G=(V,E)$ into $k$ color classes $C_1,\ldots,C_k$ such that each pair of vertices is connected by a rainbow shortest path. \lv{We will assume the predicates $\path(u,v,e_1,\ldots,e_\ell)$ and $\rainbow(e_1,\ldots,e_\ell)$ are defined precisely as in Lemma~\ref{lem:MSOstandard}.}

Let us then construct the following MSO$_2$ formula $\psi_k$:
\begin{equation*}
\begin{split}
\psi_k &\coloneqq \exists C_1,\ldots,C_k \subseteq E \Big( \forall e \in E \Big( e \in C_1 \vee \cdots \vee e \in C_k \Big) \Big) \\
	&\wedge \Big( \forall i,j\in [k], i\neq j: (C_i \cap C_j = \emptyset) \Big) \\
	&\wedge \Big( \forall u,v \in V \Big( (u \neq v) \implies \Big( \exists i\in[k] ~
	\exists e_1,\ldots,e_i \in E \big( \path(u,v,e_1,\ldots,e_i) \\
	&\wedge \rainbow(e_1,\ldots,e_i) \\
	&\wedge \forall j \in [i-1] ~ \neg \big( \exists w_1,\ldots,w_j \in E ~ \path(u,v,w_1,\ldots,w_j) \big) \big) \Big) \Big) \Big).
\end{split}
\end{equation*}
To capture the property of being a shortest path, we require there to be a $u-v$ path of length $i$, and no paths of length less than $i$. Furthermore, observe that no path of length greater than $k$ can be rainbow. The construction for $k$-\textsc{SRVC} then uses the same ideas, with the same distinctions as those specified in Lemma~\ref{lem:MSOstandard}. 
\end{proof}
}
\lv{\pfMSOstrong}

\begin{theorem}
\label{thm:twrwfpt}
Let $p\in \Nat$ be fixed. Then the problems \textsc{$k$-RC, $k$-SRC, $k$-RVC, $k$-SRVC} can be solved in time $\bigoh(n)$ on $n$-vertex graphs of treewidth at most $p$. Furthermore, \textsc{$k$-RVC, $k$-SRVC} can be solved in time $\bigoh(n^3)$ on $n$-vertex graphs of clique-width at most $p$.
\end{theorem}

\begin{proof}
The proof follows from Lemma~\ref{lem:MSOstandard} and Lemma~\ref{lem:MSOstrong} in conjunction with Fact~\ref{fact:msotreewidth} and Fact~\ref{fact:msorankwidth}.
\end{proof}

In the language of parameterized complexity~\cite{Downey2013\lv{,Niedermeier2006}}, Theorem~\ref{thm:twrwfpt} implies that these problems are fixed-parameter tractable (FPT)  parameterized by treewidth, and their vertex variants are FPT parameterized by clique-width.

\section{The Complexity of Saving Colors}
\label{sec:saving}

This section focuses on the saving versions of the rainbow coloring problems introduced in Subsection~\ref{sub:prob}, and specifically gives linear-time algorithms for $k$-\savingrc{} and $k$-\savingrvc. Our results make use of the following facts.
\begin{ourfact}[\hspace{1sp}\cite{Gottlob2007}]\label{fact:mso1connects}
There is a \emph{MSO}$_1$ predicate $\vertexconnects$ such that on a graph $G=(V,E)$ $\vertexconnects(S,u,v)$ is true iff $S \subseteq V$ is a set of vertices of $G$ such that there is a path from $u$ to $v$ that lies entirely in $S$.
\end{ourfact}
The above is easily modified to give us the following.
\begin{ourfact}\label{fact:mso2connects}
There is a \emph{MSO}$_2$ predicate $\edgeconnects$ such that on a graph $G=(V,E)$ $\edgeconnects(X,u,v)$ is true iff $X \subseteq E$ is a set of edges of $G$ such that there is path from $u$ to $v$ that lies entirely in $X$.
\end{ourfact}
\begin{theorem} \label{thm:savingrc}
For each $k\in \Nat$, the problem $k$-$\savingrc$ can be solved in time $\bigoh(n)$ on $n$-vertex graphs.
\end{theorem}
\begin{proof}
Observe that by coloring each edge of a spanning tree of $G$ with a distinct color we have that $\rc(G) \leq n-1$. Thus, if $m \geq n + k$, we have a YES-instance of $k$-$\savingrc$. Otherwise, suppose $m < n + k$. Then $G$ has a feedback edge set of size at most $k$, and hence $G$ has treewidth at most $k$. We construct a MSO$_2$ formula $\psi_k$ such that it holds that $G \models \psi_k$ is true iff $G$ is a YES-instance of $k$-$\savingrc$. Using Fact~\ref{fact:mso2connects}, we construct $\psi_k$ as follows:
\begin{equation*}
\begin{split}
\psi_k &\coloneqq \exists R_1,\ldots,R_k \subseteq E \Big( \forall i,j \in[k], i\neq j: (R_i \cap R_j = \emptyset) \Big) \\
	&\wedge \Big( \forall i \in[k]: \big( \exists e \in E (  e \in R_i ) \big) \Big) 
	\wedge |R_1\cup R_2 \cup \dots \cup R_k|\ge 2k \\
	&\wedge \Big( \forall u,v \in V \Big( (u \neq v) \implies \Big( \exists X \subseteq E \Big( \edgeconnects(X,u,v) \\
	&\wedge \forall e_1,e_2 \in X \Big( \forall i\in[k]: (e_1 \in R_i \wedge e_2 \in R_i) \implies (e_1=e_2)  \Big) \Big) \Big) \Big) \Big).
\end{split}
\end{equation*}


In the above, the expression $|A|\ge 2k$ is shorthand for the existence of $2k$ pairwise-distinct edges in $A$, which can be expressed by a simple but lengthy MSO$_2$ expression.
The formula $\psi_k$ expresses that there exist $k$ disjoint sets $R_1,\ldots,R_k$ of edges (each corresponding to a different color set with at least $1$ edge) such that their union contains at least $2k$ edges, with the following property: there is a path using at most one edge from each set $R_1,\ldots,R_k$ between every pair of vertices. Formally, this property is stated as the existence of an edge-set $X$ for each pair of vertices $u,v$ such that the graph $(V,X)$ contains an $u-v$ path that cannot repeat edges from any $R_i$.
The proof then follows by Fact~\ref{fact:msotreewidth}.
\end{proof}
To prove a similar result for $k$-$\savingrvc$, we will use the following result.
\begin{ourfact}[\hspace{1sp}\cite{Bodlaender1993}]\label{fact:twspanning}
If the treewidth of a connected graph $G$ is at least $2k^3$, then $G$ has a spanning tree with at least $k$ vertices with degree 1.
\end{ourfact}

\begin{theorem}
For each $k\in \Nat$, the problem $k$-$\savingrvc$ can be solved in time $\bigoh(n)$ on $n$-vertex graphs.
\end{theorem}
\begin{proof}
Using Fact~\ref{fact:findtw}, we will test if the treewidth of $G$ is at least $2k^3$. If it is, then by Fact~\ref{fact:twspanning} the graph $G$ has a spanning tree with at least $k$ vertices of degree 1. Each of these $k$ vertices can receive the same color, and we conclude we have a YES-instance. Otherwise, suppose the treewidth of $G$ is less than $2k^3$, and we construct a MSO$_1$ formula $\phi_k$ such that it holds that $G \models \phi_k$ is true iff $G$ is a YES-instance of $k$-$\savingrvc$. The construction is analogous to Theorem~\ref{thm:savingrc}, but instead of $\edgeconnects$ we use $\vertexconnects$ from Fact~\ref{fact:mso1connects}. The proof then follows by Fact~\ref{fact:msotreewidth}.
\end{proof}


\section{Rainbow Coloring Graphs with Small Vertex Covers}
\label{sec:vc}
In this section we turn our attention to the more general problem of determining whether the rainbow connection number is below a number specified in the input. Specifically, we show that $\rcprob$, $\rvcprob$, and $\srvcprob$ admit linear time algorithms on graphs of bounded vertex cover number. In particular, this implies that \textsc{RC, RVC, SRVC} are FPT parameterized by $\vcn(G)$.

\lv{\begin{lemma}}
\sv{\begin{lemma}[$\star$]}
\label{lem:vcvertexcolor}
Let $G=(V,E)$ be a connected graph and $p=\vcn(G)$. Then $\rvc(G)\leq 2p$ and $\srvc(G)\leq p^2$.
\end{lemma}

\newcommand{\pflemvccolor}[0]{
\begin{proof}
Let us fix a vertex cover $X$ of cardinality $p$. For the first claim, let $S$ be a spanning tree of $G$ of minimum diameter, and observe that the diameter of $S$ is at most $2p$. Let $Z$ be the non-leaf vertices of $S$. Furthermore, it must hold that $Z\cap (V\setminus X)\leq p$, and hence $|Z|\leq 2p$. Let $\alpha$ be a vertex coloring which assigns a unique color from $[|Z|]$ to each vertex in $Z$, and then assigns the color $1$ to each vertex in $V\setminus Z$. Then $\alpha$ is a rainbow vertex coloring: indeed, for any choice of $a$ and $b$, there exists an $a-b$ path whose internal vertices are a subset of $Z$.

For the second claim, consider the set $Q$ constructed as follows: for each distinct $a,b\in X$, if there exists a vertex $v$ in $V\setminus X$ adjacent to both $a$ and $b$, we choose an arbitrary such $v$ and add it into $Q$. Let $Z=Q\cup X$, and observe that $|Z|\leq p^2$. Once again, let $\alpha$ be a vertex coloring which assigns a unique color from $[|Z|]$ to each vertex in $Z$, and then assigns the color $1$ to each vertex in $V\setminus Z$. We claim that $\alpha$ is a strong rainbow vertex coloring. Indeed, consider any $a,b\in V$ and let $P$ be an arbitrary shortest $a-b$ path. Then for every internal vertex $v_i$ of $P$ such that $v_i\not \in X$, it must hold that $v_{i-1}\in X$ and $v_{i+1}\in X$. Consider the path $P'$ obtained from $P$ by replacing each internal vertex $v_i\not \in X$ by $v'_i$, where $v'_i$ is an element of $Q$ which is adjacent to $v_{i-1}$ and $v_{i+1}$. Since $P'$ has the same length as $P$ and $P'$ is rainbow colored by $\alpha$, the claim follows.
\end{proof}
}
\lv{\pflemvccolor}

The following lemma will be useful in the proof of Lemma~\ref{lem:vccolors}, a key component of our approach for dealing with $\rcprob$ on the considered graph classes. An \emph{edge separator} is an edge $e$ such that deleting $e$ separates the connected component containing $e$ into two connected components.

\lv{\begin{lemma}}
\sv{\begin{lemma}[$\star$]}
\label{lem:leaves}
Let $G=(V,E)$ be a graph and $X$ be a minimum vertex cover of $G$. Then there exist at most $2|X|-2$ edge separators which are not incident to a pendant outside of $X$.
\end{lemma}

\newcommand{\pflemleaves}[0]{
\begin{proof}
We prove the lemma by induction on $p=|X|$. If $p=1$, then the graph is a star and the lemma holds (in a star, every edge separator is incident to a pendant). 

So, assume the lemma holds for $p-1$ and assume $G$ has a vertex cover $X$ of size $p$. Let $S$ be the set of all edge separators in $G$ which are not incident to a pendant outside of $X$. If $S$ contains an edge separator $e$ whose both endpoints lie in $X$, then $e$ separates $X$ into two non-empty subsets $X_1$ and $X_2$ and every other separator has both endpoints either in $X_1$ or in $X_2$. Let $G_1$ and $G_2$ be the connected components of $G-e$ containing $X_1$ and $X_2$, respectively. Observe that $X_i$ is a vertex cover of $G_i$ for $i\in [2]$. Since $|X_1|<p$ and $|X_2|<p$, by our inductive assumption it follows that $G_1$ contains at most $2|X_1|-2$ edge separators which are not incident to a pendant outside of $X$, and similarly $G_2$ contains at most $2|X_2|-2$ edge separators which are not incident to a pendant outside of $X$. Since each edge separator in $G$ is either $e$ or an edge separator in $G_1$ or $G_2$, it follows that $|S|=1+2|X_1|-2+2|X_2|-2=1+2p-4< 2p-2$, and hence in this case the lemma holds.

On the other hand, assume $S$ contains an edge separator $e=ax$ where $x\in X, a\not \in X$. Since the connected component of $G-e$ containing $a$ is not a pendant, it follows that $a$ has a neighbor in $G$ which is different from $x$, and hence this connected component (say $G_1$) contains at least one vertex from $X$. Let $X_1=X\cap V(G_1)$, $X_2=X\setminus X_1$ and $G_2$ be the connected component of $G-e$ containing $X_2$. This implies that in this case $e$ also separates $X$ into two non-empty subsets $X_1$ and $X_2$. Furthermore, if there exists another $e'\in S$ which separates $X$ into the same sets $X_1$ and $X_2$ as $e$, then $e'$ must also be incident to $a$ and in particular this other edge $e'$ is unique; every other separator in $S$ has both endpoints either in $X_1$ or in $X_2$. Since $|X_1|<p$ and $|X_2|<p$, by our inductive assumption it follows that $G_1$ contains at most $2|X_1|-2$ edge separators which are not incident to a pendant outside of $X$, and similarly $G_2$ contains at most $2|X_2|-2$ edge separators which are not incident to a pendant outside of $X$. Since each edge separator in $G$ is either $e$ or $e'$ or an edge separator in $G_1$ or $G_2$, it follows that $|S|\leq 2+2|X_1|-2+2|X_2|-2=2+2p-4\leq 2p-2$, and hence in this case the lemma also holds.
\end{proof}
}
\lv{\pflemleaves}

For ease of presentation, we define the function $\beta$ as $\beta(p)=2p-2+p\cdot (p^2+2p\cdot 2^p)$. The next Lemma~\ref{lem:vccolors} will represent one part of our win-win strategy, as it allows us to precisely compute $\rc(G)$ when the number of edge separators is sufficiently large. We remark that an analogous claim does not hold for $\src(G)$ (regardless of the choice of $\beta$).

\lv{\begin{lemma}}
\sv{\begin{lemma}[$\star$]}
\label{lem:vccolors}
Let $G=(V,E)$ be a connected graph and $p=\vcn(G)$. Let $z$ be the number of edge separators in $G$. If $z\geq \beta(p)$, then $\rc(G)=z$.
\end{lemma}

\newcommand{\pfvccolors}[0]{
\begin{proof}
Let us fix a vertex cover $X$ of cardinality $p$. 
It is known that the number of edge separators is a lower bound for $\rc(G)$~\cite{Chartrand2010}, i.e., $\rc(G)\geq z$. We will show that $z$ is also an upper bound for $\rc(G)$.

Consider the edge $z$-coloring $\alpha$ constructed as follows. Since $X$ is a vertex cover and, by Lemma~\ref{lem:leaves} in conjunction with our assumption on $z$, there are at least $p\cdot (p^2+(4+p)\cdot 2^p)$ leaves in $G$, it follows that there must exist some $x\in X$ adjacent to at least $z'=p^2+2p\cdot 2^p$ pendants. Let $\{e_1,\dots,e_{z'}\}$ be the edges incident to both $x$ and a pendant vertex, and let $\{e_{z'+1},\dots,e_z\}$ be all the remaining edge separators; then for each edge separator we set $\alpha(e_i)=i$.

Let $f_1,\ldots,f_q$ be the edges of $G[X]$ which are not edge separators; for each such edge we set $\alpha(f_i)=z'-i$. Observe that for each such $f_i$ we have $\alpha(f_i)>2p\cdot 2^p$.

Consider the set $\tau=\SB T_i \SM T_i\text{ is a type in } G \text{ and }|N(T_i)|>1 \SE$. Let $Q_i=\{2pi+1,\dots,2pi+2p\}$. For each $T_i\in \tau$, we let $G_i$ be the subgraph of $G$ on $T_i\cup N(T_i)$ which contains exactly the edges incident to $T_i$. Then $G_i$ is bipartite, and furthermore can be rainbow colored using at most $2p$ colors as follows: we pick an arbitrary $y\in T_i$ and uniquely color all edges in $G_i$ incident to $y$ using colors $c_1,\dots,c_p$, and for every other vertex in $T_i$ we color all edges in $G_i$ incident to $y'$ using colors $c_{1+p},\dots,c_{2p}$. For each type $T_i\in \tau$, we let $\alpha$ color the edges incident to $T_i$ in this manner using the colors from $Q_i$.

We will proceed by arguing that $\alpha$ is a rainbow edge $z$-coloring of $G$, but before that we make three key observations. First, there are only two cases when $\alpha$ can use the same color for two distinct edges $e,f$: either one of $e,f$ is an edge between $x$ and a pendant, or both $e,f$ occur in some $G_i$. Second, if $|T_i|>1$, then for every $u\in N(T_i)$ and every $v\in V(G_i)$ and every color $c$, there exists a rainbow $u-v$ path in $G_i$ under $\alpha$ which does not use $c$. Third, each $G_i$ is rainbow colored by $\alpha$.

We now make the following case distinction.
\begin{enumerate}
\item Let $a,b\in V$ be such that neither is a pendant adjacent to $x$. Consider an arbitrary $a-b$ path $P$ such that the number of pairs of edges in $P$ assigned the same color by $\alpha$ is minimized. If $P$ contains two edges $e,f$ such that $\alpha(e)=\alpha(f)$, then both $e$ and $f$ must occur in some $G_i$. Let $t$ and $u$ be the first and last vertex in $V(G_i)$ on $P$, respectively. Since $G_i$ is rainbow colored by $\alpha$, there exists a $t-u$ rainbow path $P^*$ in $G_i$. Let $P'$ be obtained from $P$ by replacing the path segment between $t$ and $u$ by $P^*$; by the key observation made above, it follows that $P'$ has a strictly lower number of pairs of edges in $P$ which are assigned the same color by $\alpha$, hence contradicting our choice of $P$.
\item Let $a$ be a pendant adjacent to $x$, and $b\in V$. Let $c=\alpha(xa)$. Consider an arbitrary $a-b$ path $P$ such that the number of pairs of edges in $P$ assigned the same color by $\alpha$ is minimized. If $P$ contains two edges $e,f$ such that $\alpha(e)=\alpha(f)\neq c$, then both $e$ and $f$ must occur in some $G_i$ s.t. $|T_i|>1$. Let $t$ and $u$ be the first and last vertex in $V(G_i)$ on $P$, respectively. Since $t\in X\cap N(T_i)$, by our observations above it follows that there exists a rainbow $t-u$ path $P^*$ in $G_i$ which avoids $c$. Hence the path obtained by replacing the path segment between $t$ and $u$ by $P^*$ once again contradicts our choice of $P$. On the other hand, if $P$ contains an edge $e$ such that $\alpha(e)=c$, then either $e$ is an edge in $G[X]$ or $e$ is incident to some $T_i$. In the latter case, the same argument can be used to contradict our choice of $P$. In the former case it follows by construction of $\alpha$ that $c$ only occurs on the edge $(x,a)$ and on $e$, and furthermore $e$ is contained in some $2$-edge-connected component $D$ of $G$. Let $d,w$ be the first and last vertex, respectively, in $D$ which occurs in $P$, and let $P'$ be the path obtained from $P$ by replacing the path segment between $d$ and $w$ by an arbitrary rainbow path segment in $D$ which does not contain $e$. It is readily verified that the colors which occur in $D$ are only repeated on edges between $x$ and pendants, and in particular such edges cannot occur on $P'$. Hence $P'$ again contradicts our choice of $P$.
\end{enumerate}

To summarize, for any $a,b\in V$ there exists a rainbow $a-b$ path under $\alpha$, and hence $\alpha$ witnesses that $\rc(G)\leq z$. We conclude that $\rc(G)=z$.
\end{proof}
}
\lv{\pfvccolors}
\newcommand{\pfvccolorsshort}{
\begin{proof}[Sketch]
Let us fix a vertex cover $X$ of cardinality $p$. 
It is known that the number of edge separators is a lower bound for $\rc(G)$~\cite{Chartrand2010}, i.e., $\rc(G)\geq z$. We will show that $z$ is also an upper bound for $\rc(G)$.

Consider the edge $z$-coloring $\alpha$ constructed as follows. Since $X$ is a vertex cover and, by Lemma~\ref{lem:leaves} in conjunction with our assumption on $z$, there are at least $p\cdot (p^2+(4+p)\cdot 2^p)$ leaves in $G$, it follows that there must exist some $x\in X$ adjacent to at least $z'=p^2+2p\cdot 2^p$ pendants. Let $\{e_1,\dots,e_{z'}\}$ be the edges incident to both $x$ and a pendant vertex, and let $\{e_{z'+1},\dots,e_z\}$ be all the remaining edge separators; then for each edge separator we set $\alpha(e_i)=i$.

Let $f_1,\ldots,f_q$ be the edges of $G[X]$ which are not edge separators; for each such edge we set $\alpha(f_i)=z'-i$. Observe that for each such $f_i$ we have $\alpha(f_i)>2p\cdot 2^p$.

Consider the set $\tau=\SB T_i \SM T_i\text{ is a type in } G \text{ and }|N(T_i)|>1 \SE$. Let $Q_i=\{2pi+1,\dots,2pi+2p\}$. For each $T_i\in \tau$, we let $G_i$ be the subgraph of $G$ on $T_i\cup N(T_i)$ which contains exactly the edges incident to $T_i$. Then $G_i$ is bipartite, and furthermore can be rainbow colored using at most $2p$ colors as follows: we pick an arbitrary $y\in T_i$ and uniquely color all edges in $G_i$ incident to $y$ using colors $c_1,\dots,c_p$, and for every other vertex in $T_i$ we color all edges in $G_i$ incident to $y'$ using colors $c_{1+p},\dots,c_{2p}$. For each type $T_i\in \tau$, we let $\alpha$ color the edges incident to $T_i$ in this manner using the colors from $Q_i$.

It is possible to verify that $\alpha$ is a rainbow edge $z$-coloring of $G$, and hence $\rc(G)\leq z$. We conclude $\rc(G)=z$.
\end{proof}
}

\lv{\begin{lemma}}
\sv{\begin{lemma}[$\star$]}
\label{lem:vcrcbound}
Let $G=(V,E)$ be a graph with a vertex cover $X\subseteq V$ of cardinality $p$. Let $z$ be the number of edge separators in $G$. If $z< \beta(p)$, then $\rc(G)\leq \beta(p)+p^2+2^p\cdot 2p$.
\end{lemma}

\newcommand{\pflemvcrcbound}[0]{
\begin{proof}
Consider the following edge coloring $\alpha$ which assigns a unique color to each edge in $G[X]$ and to each edge incident to a pendant. For each nonempty type $T_i$, we choose an arbitrary vertex $y_i$ and let $\alpha$ assign a unique color for each of the at most $p$ edges incident to $y_i$. Finally, for each type $T_i$ and each $x\in X$ adjacent to (the vertices of) $T_i$, $\alpha$ uses a single new color for all edges between $x$ and the vertices in $T_i$. It is readily verified that $\alpha$ uses no more than $z+p^2+2^p\cdot 2p$ colors.

We argue that $\alpha$ is rainbow. Let $G_i$ be the subgraph of $G$ on $T_i\cup N(T_i)$ which contains exactly the edges incident to $T_i$, and observe that each $G_i$ is rainbow colored by $\alpha$. Consider any $a,b\in V$ and let $P$ be an $a-b$ path such that the number of pairs of edges in $P$ assigned the same color by $\alpha$ is minimized. By construction of $\alpha$, two edges $e,f$ in $P$ may only have the same color if $e,f$ are both incident to some $T_i$. Let $t$ and $u$ be the first and last vertex in $V(G_i)$ on $P$, respectively. Since $G_i$ is rainbow colored by $\alpha$, there exists a $t-u$ rainbow path $P^*$ in $G_i$ under $\alpha$. Let $P'$ be obtained from $P$ by replacing the path segment between $t$ and $u$ by $P^*$. Then $P'$ has a strictly lower number of pairs of edges in $P$ with the same color, which contradicts our choice of $P$.
\end{proof}}

\lv{\pflemvcrcbound}

\begin{theorem}
\label{thm:vcnfpt}
Let $p\in \Nat$ be fixed. Then the problems \textsc{RC, RVC, SRVC} can be solved in time $\bigoh(n)$ on $n$-vertex graphs of vertex cover number at most $p$.
\end{theorem}

\begin{proof}
For \textsc{RVC} and \textsc{SRVC}, we first observe that if $k$ (the queried upper bound on the number of colors) is greater than $2p$ and $p^2$, respectively, then the algorithm can immediately output YES by Lemma~\ref{lem:vcvertexcolor}. Otherwise we use Theorem~\ref{thm:twrwfpt} and the fact that the vertex cover number is an upper bound on the treewidth to compute a solution in $\bigoh(n)$ time.

For \textsc{RC}, it is well known that the total number of edge separators in $G$, say $z$, can be computed in linear time on graphs of bounded treewidth. If $z\geq \beta(p)$, then by Lemma~\ref{lem:vccolors} we can correctly output YES when $z\leq k$ and NO when $z> k$. On the other hand, if $z<\beta(p)$, then by Lemma~\ref{lem:vcrcbound} the value $\rc(G)$ is upper-bounded by a function of $p$. We compare $k$ and this upper bound on $\rc(G)$; if $k$ exceeds the upper bound on $\rc(G)$, then we output YES, and otherwise we can use Theorem~\ref{thm:twrwfpt} along with the fact that the vertex cover number is an upper bound on the treewidth to compute a solution in $\bigoh(n)$ time.
\end{proof}

\section{Concluding Notes}

We presented new positive and negative results for the most prominent variants of  rainbow coloring. We believe that the techniques presented above, and in particular the win-win approaches used in Section~\ref{sec:saving} and Section~\ref{sec:vc}, can be of use also for other challenging connectivity problems. 

It is worth noting that our results in Section~\ref{sec:mso} leave open the question of whether \textsc{Rainbow Coloring} or its variants can be solved in (uniformly) polynomial time on graphs of bounded treewidth. Hardness results for related problems~\cite{Uchizawa13,Lauri15} do not imply that finding an optimal coloring of a bounded-treewidth graph is hard, and it seems that new insights are needed to determine the complexity of these problems on graphs of bounded treewidth.

\bibliographystyle{abbrv}
\bibliography{bibliography}

\newpage

\sv{
\appendix

\section{Proofs for Section~\ref{sec:hardness_srvc}}

\subsection{Proof of Lemma~\ref{lem:ssrvc}}
\pfssrvc

\subsection{Proof of Theorem~\ref{thm_srvc_hardness}}
\pfsrvc

\section{Proofs and Omissions for Section~\ref{sec:hardness_srvc}}

\subsection{Omitted Formulas}

\begin{equation*}
\begin{split}
\path(u,v,e_1,\ldots,e_\ell) &\coloneqq \exists v_1,\ldots,v_{\ell-1} \in V \\
	&\Big( \forall i,j \in [\ell-1], i\neq j: (v_i \neq v_j) \Big) \\
	&\wedge \inc(e_1,u) \wedge \inc(e_\ell,v)\\
&\wedge \Big( \forall i\in[\ell-1]: (\inc(e_i,v_i) \wedge \inc(e_{i+1},v_i)) \Big).
\end{split}
\end{equation*}

\begin{equation*}
\begin{split}
\rainbow(e_1,\ldots,e_\ell) &\coloneqq \forall i\in[\ell] ~ \exists j\in[k]: \Big( e_i \in C_j 
	\wedge ( \forall p\neq i: e_p \notin C_j ) \Big).
\end{split}
\end{equation*}

\subsection{Proof of Lemma~\ref{lem:MSOstrong}}
\pfMSOstrong

\section{Proofs for Section~\ref{sec:vc}}

\subsection{Proof of Lemma~\ref{lem:vcvertexcolor}}
\pflemvccolor

\subsection{Proof of Lemma~\ref{lem:leaves}}
\pflemleaves

\subsection{Proof of Lemma~\ref{lem:vccolors}}
\pfvccolors

\subsection{Proof of Lemma~\ref{lem:vcrcbound}}
\pflemvcrcbound

}

\end{document}